\documentclass[a4paper,english]{article}

\usepackage[english]{babel}
\usepackage{german}
\usepackage{amssymb,amsfonts,amsmath,amsthm,cite}
\usepackage[latin1]{inputenc}
\usepackage{fullpage}
\usepackage{graphicx}
\usepackage{subfig}
\usepackage{float}
\usepackage{color,hyperref,enumerate}

\usepackage[boxed]{algorithm2e}
\newtheorem*{observation*}{Observation}
\newtheorem*{openprob*}{Open Problem}

\newtheorem*{conjecture*}{Conjecture}
\newtheorem*{example*}{Example}
\newcommand{\SE}{\textsc{StarExp($k$)}}
\newcommand{\MSE}{\textsc{MaxStarExp($k$)}}

\SetKwInput{KwData}{Input}
\SetKwInput{KwResult}{Output}
\newtheorem{theorem}{Theorem}

\newtheorem{corollary}{Corollary}

\newtheorem{definition}{Definition}

\newtheorem{lemma}{Lemma}

\begin{document}
	
\title{\vspace{-0.5cm}The temporal explorer who returns to the base\thanks{This work was partially supported 
		by NeST initiative of the School of EEE and CS at the University of Liverpool and by the EPSRC Grants 
		EP/P020372/1 and EP/P02002X/1.}}
\author{Eleni C.~Akrida\thanks{Department of Computer Science, University of Liverpool, Liverpool, UK. 
		Email: \texttt{e.akrida@liverpool.ac.uk}} 
	\and George B.~Mertzios\thanks{Department of Computer Science, Durham University, Durham, UK. 
		Email: \texttt{george.mertzios@durham.ac.uk}} 
	\and Paul G.~Spirakis\thanks{Department of Computer Science, University of Liverpool, Liverpool, UK, and 
		Department of Computer Engineering \& Informatics, University of Patras, Greece.
		Email: \texttt{p.spirakis@liverpool.ac.uk}}
	}
\date{\vspace{-1.0cm}}	
\maketitle

\begin{abstract}
In this paper we study the problem of exploring a temporal graph (i.e.~a graph that changes over time), 
in the fundamental case where the underlying static graph is a star. The aim of the exploration problem in a temporal star is to find a temporal walk which starts at the center of the star, visits all leafs, and eventually returns back to the center. 
We initiate a systematic study of the computational complexity of this problem, depending on the number $k$ of time-labels that every edge is allowed to have; that is, on the number $k$ of time points where every edge can be present in the graph.
To do so, we distinguish between the decision version \SE\, asking whether a complete exploration of the instance exists, and the maximization version \MSE\ of the problem, asking for an exploration schedule of the greatest possible number of edges in the star.
We present here a collection of results establishing the computational complexity of these two problems. 
On the one hand, we show that both \textsc{MaxStarExp(2)} and \textsc{StarExp(3)} can be efficiently solved 
in $O(n\log n)$ time on a temporal star with $n$ vertices. 
On the other hand, we show that, for every $k\geq 6$, 
\SE\ is NP-complete and \MSE\ is APX-hard, and thus it does not admit a PTAS, unless P~=~NP. 
The latter result is complemented by a polynomial-time 2-approximation algorithm for \MSE, for every $k$, 
thus proving that \MSE\ is APX-complete. 
Finally, we give a partial characterization of the classes of temporal stars with random labels 
which are, asymptotically almost surely, yes-instances and no-instances for \SE\, respectively.

\end{abstract}

\section{Introduction and motivation}\label{sec:intro}

A temporal graph is, roughly speaking, a graph that changes over time. Several networks, both modern and traditional, including social networks, transportation networks, information and communication networks, can be modeled as temporal graphs. The common characteristic in all the above examples is that the network structure, i.e.~the underlying graph topology, is subject to discrete changes over time. Temporal graphs naturally model such time-varying networks using time-labels on the edges of a graph to indicate moments of existence of those edges, while the vertex set remains unchanged. This formalism originates in the foundational work of Kempe et al.~\cite{kempe}.

In this work, we focus in particular on temporal graphs where the underlying graph is a star graph and we consider the problem of exploring such a temporal graph starting and finishing at the center of the star. 
The motivation behind this is inspired from the well known Traveling Salesman Problem (TSP). The latter asks the following question: ``Given a list of cities and the distances between each pair of cities, what is the shortest possible route that visits each city and returns to the origin one?''. 
In other words, given an undirected graph with edge weights where vertices represent cities and edges represent the corresponding distances, find a minimum-cost Hamiltonian cycle. However, what happens when the traveling salesman has particular temporal constraints that need to be satisfied, 
e.g.~(s)he can only go from city $A$ to city $B$ on Mondays or Tuesdays, or when (s)he needs to take the train and, hence, schedule his/her visit based on the train timetables? 
In particular, consider a traveling salesman who, starting from his/her home town, has to visit $n-1$ other towns via train, always \emph{returning to their own home town} after visiting each city. There are trains between each town and the home town only on specific times/days, possibly different for different towns, and the salesman knows those times in advance. 
Can the salesman decide whether (s)he can visit all towns and return to the own  home town by a certain day?

\medskip

\textbf{Previous work.} 
Recent years have seen a growing interest in dynamic network studies. Due to its vast applicability in many areas, the notion of temporal graphs has been studied from different perspectives under various names such as time-varying~\cite{krizanc1,flocchiniMS09,TangMML10-ACM}, evolving~\cite{xuan,clementi,Ferreira-MANETS-04}, dynamic~\cite{GiakkoupisSS14}, and graphs over time~\cite{Leskovec-Kleinberg-Faloutsos07}; for a recent attempt to integrate existing models, concepts, and results from the distributed computing perspective see the survey papers~\cite{flocchini1,flocchini2,CasteigtsFloccini12} and the references therein.
Temporal data analytics, temporal flows, as well as various temporal analogues of known static graph concepts such as cliques, vertex covers, diameter, distance, connectivity and centrality have also been studied~\cite{gionis2,gionis1,akrida,akridaTOCS,akridaAlgosensors,akridaCIAC,akridaICALP,spirakis,neidermeier,viardClique,viardCliqueTCS}.

Notably, temporal graph exploration has been studied before~\cite{erlebach,michailTSP}; Erlebach et al.~\cite{erlebach} define the problem of computing a foremost exploration of all vertices in a temporal graph (\textsc{Texp}), without the requirement of returning to the starting vertex. They show that it is NP-hard to approximate \textsc{Texp} 
with ratio $O(n^{1-\varepsilon})$ for any $\varepsilon>0$, and give explicit construction of graphs that need $\Theta(n^2)$ steps for \textsc{Texp}. They also consider special classes of underlying graphs, such as the grid, as well as the case of random temporal graphs where edges appear in every step with independent probabilities. Michail and Spirakis~\cite{michailTSP} study a temporal analogue of TSP(1,2) where the objective is to explore the vertices of a complete directed temporal graph with edge weights from $\{1,2\}$ with the minimum total cost. 

We focus here on the exploration of temporal stars, inspired by the \textsc{Traveling Salesman} paradigm where the salesman returns to his base after visiting every city. The \textsc{Traveling Salesman Problem} is one of the most well-known combinatorial optimization problems, which still poses great challenges despite having been intensively studied for more than sixty years. For the \textsc{Symmetric TSP}, where the given graph is undirected (as is the case for the temporal version of the problem that we consider here) and the edge costs obey the triangle inequality, the best known approximation algorithm is still the celebrated $3/2$ of Christofides~\cite{Christofides76}, despite forty years of intensive efforts to improve it. Only recently, Gharan et al.~\cite{Gharan11} proved that the \textsc{Graphic TSP} special case where the costs correspond to shortest path distances of some given graph can be approximated within $3/2\varepsilon$, for a small constant $\varepsilon>0$ (which was further improved by subsequent works, e.g.~\cite{Sebo14}). For the \textsc{Asymmetric TSP} where paths may not exist in both directions or the distances might be different depending on the direction, the $O(\log{n})$-approximation of~\cite{Frieze82} was the best known for almost three decades, improved only recently to $O(\log{n}/ \log{\log{n}})$~\cite{Asadpour17}. Online TSP-related problems as well as versions of TSP where each node must be visited within a given time window, have also been recently studied~\cite{azar2016,paulsen2015}.

\medskip

\textbf{The model and definitions.}
It is generally accepted to describe a network topology using a graph, the vertices and edges of which represent the communicating entities and the communication opportunities between them, respectively. Unless otherwise stated, we denote by $n$ and $m$ the number of vertices and edges of the graph, respectively.
We consider graphs whose edge availabilities are described by sets of positive integers (labels), one set per edge.
\begin{definition}[Temporal Graph]
	Let $G=(V,E)$ be a graph. A temporal graph on $G$ is a pair $(G,L)$, where $L: E \to 2^{\mathbb{N}}$ is a time-labeling function, called a \emph{labeling} of $G$, which assigns to every edge of $G$ a set of discrete-time labels. The labels of an edge are the \emph{discrete time instances} at which it is available.
\end{definition}

More specifically, we focus on temporal graphs whose underlying graph is an undirected star, i.e.~a connected graph of $m=n-1$ edges which has $n-1$ leaves, i.e.~vertices of degree~$1$.
\begin{definition}[Temporal Star]
	A temporal star is a temporal graph $(G_s,L)$ on a star graph $G_s = (V,E)$. Henceforth, we denote by $c$ the center of $G_s$, i.e.~the vertex of degree $n-1$.
\end{definition}

\begin{definition}[Time edge]
	Let $e=\{u,v\}$ be an edge of the underlying graph of a temporal graph and consider a label $l\in L(e)$. The ordered triplet $(u,v,l)$ is called \emph{time edge}.\footnote{Note that an undirected edge $e=\{u,v\}$ is associated with $2\cdot |L(e)|$ time edges, namely both $(u,v,l)$ and $(v,u,l)$ for every $l\in L(e)$.}
\end{definition}

A basic assumption that we follow here is that when a message or an entity passes through an available link at time $t$, then it can pass through a subsequent link only at some time $t'>t$ and only at a time at which that link is available.

\begin{definition}[Journey]
	A \emph{temporal path} or \emph{journey} $j$ from a vertex $u$ to a vertex $v$ \emph{($(u, v)$-journey)} is a sequence of time edges $(u, u_1, l_1)$, $(u_1, u_2, l_2)$, $\ldots$ , $(u_{k-1}, v, l_k)$, such that $l_i < l_{i +1}$, for each $1 \leq i \leq k - 1$. We call the last time label, $l_k$, \emph{arrival time} of the journey.
\end{definition}

Given a temporal star $(G_s,L)$, on the one hand we investigate the complexity of deciding whether $G_s$ is \emph{explorable}: we say that $(G_s,L)$ is explorable if there is a journey starting and ending at the center of $G_s$ that visits every node of $G_s$. Equivalently, we say that there is an \emph{exploration} that \emph{visits} every node, and \emph{explores} every edge, of $G_s$. 
On the other hand, we investigate the complexity of computing an exploration schedule that visits the greatest number of edges. 
A (partial) exploration of a temporal star is a journey $J$ that starts and ends at the center of $G_s$ which visits some nodes of $G_s$; its size $|J|$ is the number of nodes of~$G_s$ that are visited by $J$. 
We, therefore, identify the following problems:

\vspace{0cm} \noindent \fbox{ 
	\begin{minipage}{0.96\textwidth}
		\begin{tabular*}{\textwidth}{@{\extracolsep{\fill}}lr} \SE \ & \\ \end{tabular*}
		
		\vspace{1.2mm}
		{\bf{Input:}}  A temporal star $(G_s,L)$ such that every edge has at most $k$ labels.\\
		{\bf{Question:}} Is $(G_s,L)$ explorable?
\end{minipage}} \vspace{0,3cm}

\vspace{0cm} \noindent \fbox{ 
	\begin{minipage}{0.96\textwidth}
		\begin{tabular*}{\textwidth}{@{\extracolsep{\fill}}lr} \MSE \ & \\ \end{tabular*}
		
		\vspace{1.2mm}
		{\bf{Input:}}  A temporal star $(G_s,L)$ such that every edge has at most $k$ labels.\\
		{\bf{Output:}} A (partial) exploration of $(G_s,L)$ of \emph{maximum} size.
\end{minipage}} \vspace{0,3cm}

Note that the case where one edge $e$ of the input temporal star has only one label is degenerate.
Indeed, in the decision variant (i.e.~\SE) we can immediately conclude that $(G_s,L)$ is a no-instance as this edge cannot be explored; similarly, in the maximization version (i.e.~\MSE) we can just ignore edge $e$ for the same reason. 
We say that we ``enter'' an edge $e=\{c,v\}$ of $(G_s,L)$ when we cross the edge from $c$ to $v$ at a time on which the edge is available. We say that we ``exit'' $e$ when we cross it from $v$ to $c$ at a time on which the edge is available.
Without loss of generality we can assume that, in an exploration of $(G_s,L)$, the entry to any edge $e$ is followed by the exit from $e$ at the earliest possible time. That is, if the labels of an edge $e$ are 
$l_1, l_2, \ldots, l_k$ and we enter $e$ at time $l_i$, we exit at time $l_{i+1}$. The reason is that, waiting at a leaf (instead of exiting as soon as possible) does not help in exploring more edges; 
we are better off returning to the center $c$ as soon as possible.

\medskip

\textbf{Our contribution.}
In this paper we initiate a systematic study of the computational complexity landscape of the temporal star exploration problems \SE\ and \MSE, depending on the maximum number $k$ of labels allowed per edge. 
As a warm-up, we first prove in Section~\ref{sec:2labels} that the maximization problem~\textsc{MaxStarExp(2)}, 
i.e.~when every edge has two labels per edge, can be efficiently solved in $O(n\log n)$ time; 
sorting the labels of the edges is the dominant part in the running time.
When every edge is allowed to have up to three labels, the situation becomes more interesting. 
The decision problem \textsc{StarExp(3)} can be easily reduced to an equivalent \textsc{2SAT} instance 
with $O(n^2)$ clauses. In Section~\ref{sec:3labels}, we prove that \textsc{StarExp(3)} can be solved in $O(n\log n)$ time. To do so, we provide a more sophisticated algorithm that, given an instance of 
temporal star exploration with $n$ vertices, constructs an equivalent \textsc{2SAT} instance with $O(n)$ 
clauses. Unfortunately, this approach does not extend to the maximization problem \textsc{MaxStarExp(3)}, whose time complexity remains open.

In Section~\ref{sec:L-reduction} we prove that, for every $k\geq 6$, 
the decision problem \SE\ is NP-complete and the maximization problem \MSE\ is APX-hard, and thus it does not admit a Polynomial-Time Approximation Scheme (PTAS), unless P~=~NP. This is proved by a reduction from a special case of \textsc{3SAT}, namely \textsc{3SAT(3)}, where every variable appears in at most three clauses. 
We complement these hardness results by providing, for every~$k$, a greedy 2-approximation algorithm for \MSE\ in Section~\ref{greedy-sec}, thus proving that \MSE\ is APX-complete for $k\geq 6$.
Finally, in Section~\ref{sec:random} we study the problem of exploring a temporal star whose edges have $k$ random labels (chosen uniformly at random within an interval $[1,\alpha]$, for some $\alpha\in\mathbb{N}$). 
We partially characterize the classes of temporal stars which, asymptotically almost surely, admit a complete (resp.~admit no complete) exploration.

\section{Efficient optimization algorithm for two labels per edge}\label{sec:2labels}

In this section we show that, when every edge has two labels, a maximum size exploration in $(G_s,L)$ can be efficiently solved in $O(n \log{n})$ time. Thus, clearly, the decision variation of the problem, i.e.~\textsc{StarExp(2)}, can also be solved within the same time bound.

\begin{theorem}
	\textsc{MaxStarExp(2)} can be solved in $O(n \log{n})$ time.
\end{theorem}
\begin{proof}
	We show that \textsc{MaxStarExp(2)} is reducible to the Interval Scheduling Maximization Problem (ISMP). 
	
	\vspace{0cm} \noindent \fbox{ 
		\begin{minipage}{0.96\textwidth}
			\begin{tabular*}{\textwidth}{@{\extracolsep{\fill}}lr} Interval Scheduling Maximization Problem \  (ISMP)& \\ \end{tabular*}
			
			\vspace{1.2mm}
			{\bf{Input:}}  A set of intervals, each with a start and a finish time.\\
			{\bf{Output:}} Find a set of non-overlapping intervals of maximum size.
	\end{minipage}} \vspace{0,3cm}
	
	Every edge $e$ of $(G_s,L)$ with labels $l_e<l_e^*$ can be viewed as an interval to be scheduled that has start time $l_e$ and finish time $l_e^*+0.5$; indeed, exiting $e$ using label $l_e^*$ means that we return to the centre of $G_s$ at the ``end'' of day $l_e^*$ and, thus, can subsequently explore other edges only at days \emph{after} $l_e^*$. So, to avoid scheduling/exploring an edge $e'$ that also has a label equal to $l_e^*$ using that label, we must add a positive number $\alpha \in (0,1)$, e.g.~$\alpha = 0.5$, to the label $l_e^*$ of $e$ as well as to the largest label of every other edge.
	
	So, given $(G_s,L)$ we construct a set of $n-1$ intervals as follows: for every edge $e\in E$, we create an interval $I_e$ with start time $l_e$ and finish time $l_e^*+0.5$. We say that two edges are conflicting when their corresponding intervals are overlapping. Clearly, any (partial) exploration of $(G_s,L)$ corresponds to a set of non-overlapping intervals of the same size as the exploration, and vice versa. 
	
	The following greedy algorithm finds an optimal solution for ISMP~\cite{algorithm_design} and can therefore find the optimal solution for \textsc{MaxStarExp(2)}:
	\begin{enumerate}
		\item Start with the set $S=E$ of all edges. Select the edge, $e$, with the smallest largest label (equivalent to the earliest finish time or the corresponding interval).
		\item Remove from $S$ the edge $e$ and all conflicting edges.
		\item Repeat until $S$ is empty.
	\end{enumerate}
	
	The above works in $(|E|\log{|E|}) = O(n \log{n})$ time.
\end{proof}

\section{Efficiently deciding exploration with three labels per edge}\label{sec:3labels}

In this section we show that, when every edge has up to three labels, the decision problem of whether a 
complete exploration of the temporal star exists (i.e.~\textsc{StarExp(3)}), can be efficiently solved in $O(n\log n)$ time.

Before we present our $O(n\log n)$-time algorithm, we first outline here an easy $O(n^2)$-time algorithm that decides~\textsc{StarExp(3)}. 
To this end, first note that we can easily deal with all edges $e$ that have exactly two labels; in this case, $e$ must be explored by entering at the smallest and leaving at the largest label. Thus the instance can be reduced to a smaller one, with only edges with three labels, by removing all labels from other edges which are conflicting with the exploration of~$e$.\footnote{Assume that the two labels of the edge $e$ are $l_1$ and $l_2$, where $l_1<l_2$. 
Then, a label $l'$ of another edge $e'$ is conflicting with the exploration of $e$ if $l_1\leq l' \leq l_2$. If, after removing all labels of other edges which are conflicting with the exploration of $e$, an edge $e'$ remains with only one label, or with two labels $l'_1,l'_2$ such that $l'_1<l_1<l_2<l'_2$, then the exploration of both $e$ and $e'$ is not possible, and thus the instance is a no-instance.} 
Furthermore, as mentioned above, we can assume without loss of generality that, in an exploration of $(G_s,L)$, the entry to any edge $e$ is followed by the exit from $e$ at the earliest possible time. 
We now reduce the problem to \textsc{2SAT} as follows. For every edge $e_i$ with labels $l_{i,1} < l_{i,2} < l_{i,3}$, we define the two possible exploration windows for this edge, namely $[l_{i,1}, l_{i,2}]$ and $[l_{i,2} , l_{i,3}]$. Furthermore we assign to edge $e_i$ a Boolean variable $x_i$ such that the truth assignment $x_i=0$ (resp.~$x_i=1$) means that edge $e_i$ is explored in the interval $[l_{i,1}, l_{i,2}]$ (resp.~$[l_{i,2} , l_{i,3}]$).
Using these variables, we create a number of 2-clauses as follows. 
For any two edges $e_i$ and $e_j$, if the exploration of $e_i$ using its first (resp.~second) exploration window is conflicting with the exploration of $e_j$ using its first (resp.~second) exploration window, we add the clause $({x_i} \vee {x_j})$ (resp.~$(\neg{x_i} \vee \neg{x_j})$). 
Similarly, if the exploration of $e_i$ using its second (resp.~first) exploration window is conflicting with the exploration of $e_j$ using its first (resp.~second) exploration window, 
we add the clause $(\neg{x_i} \vee {x_j})$ (resp.~$({x_i} \vee \neg{x_j})$). 
The constructed 2-CNF formula is satisfiable if and only if $(G_s,L)$ is explorable. Furthermore this formula contains $O(n^2)$ clauses in total, and thus the exploration problem can be solved in $O(n^2)$ time using a linear-time algorithm for \textsc{2SAT}~\cite{even}.

In the next theorem we prove that \textsc{StarExp(3)} can be reduced to \textsc{2SAT} such that the number of clauses in the constructed formula is linear in $n$. For simplicity of the presentation, we assume in the next theorem that all labels in the input are different; we later prove in Corollary~\ref{decision-3-labels-general-cor} that this assumption can be actually removed, thus implying an $O(n\log{n})$-time algorithm on general input instances with at most three labels per edge.

\begin{theorem}
	\textsc{StarExp(3)} can be solved in $O(n\log{n})$ time on instances with \emph{distinct} labels.
\end{theorem}\label{thm:lognStarExp3}
\begin{proof}
	Consider an instance $(G_s,L)$ of \textsc{StarExp(3)} with distinct labels. To prove the statement, we consider a variable $x_e$ for each $e\in E$. Setting $x_e=0$ will be associated with exploring $e$ using its first exploration window, and setting $x_e=1$ will be associated with exploring $e$ using its second exploration window. We will show below how to construct a 2-SAT formula, $F$, of size linear to the input, which is satisfiable if and only if $(G_s,L)$ is explorable. Let $a_e, b_e,c_e$ be the three labels of $e$, and let $I_e$ denote the time interval $[a_e,c_e]$. Also, let $I_e^1$ denote the first exploration window, $[a_e,b_e]$, of $e$, and $I_e^2$ denote the second exploration window, $[b_e,c_e]$, of $e$.
	
	First, notice that for any two edges $e_1$ and $e_2$, there are only $3$ cases up to renaming:
	\begin{enumerate}
		\item\label{item:a} $b_{e_1} \in I_{e_2}$. Here, there are two sub-cases:
		\begin{enumerate}
			\item\label{item:a1} $b_{e_1} \in I_{e_2}^1$. See, for example, Figure~\ref{fig:case11}.
			\item\label{item:a2} $b_{e_1} \in I_{e_2}^2$. See, for example, Figure~\ref{fig:case12}.
		\end{enumerate}
		\item\label{item:b} $c_{e_1} \in I_{e_2}^1$. See, for example, Figure~\ref{fig:case2}. Here, we do not consider the case $c_{e_1} \in I_{e_2}^2$, as in that case we would have $b_{e_2} \in I_{e_1}$ which reduces to case~\ref{item:a}.
		\item $I_{e_1}$ and $I_{e_2}$ do not overlap.
	\end{enumerate}
	\begin{figure}[h]
		\centering
		\subfloat[First case (1a) where the middle label of an edge lies within the exploration window of another edge.\label{fig:case11}]{%
			\includegraphics[width=0.26\linewidth]{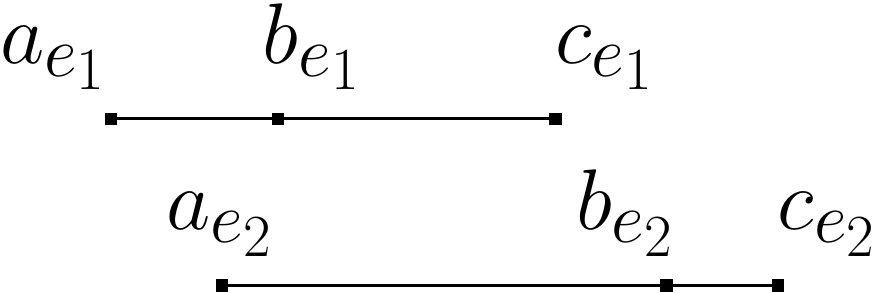}
		}
		\hfill
		\subfloat[Second case (1b) where the middle label of an edge lies within the exploration window of another edge.\label{fig:case12}]{%
			\includegraphics[width=0.26\linewidth]{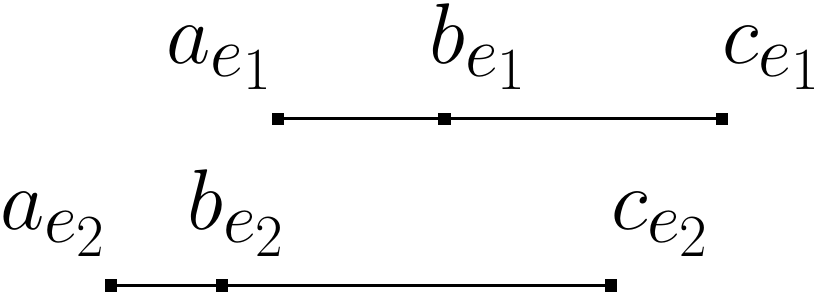}
		}
		\hfill
		\subfloat[Case (2) where the largest label of an edge lies within the first exploration window of another edge.\label{fig:case2}]{%
			\includegraphics[width=0.26\linewidth]{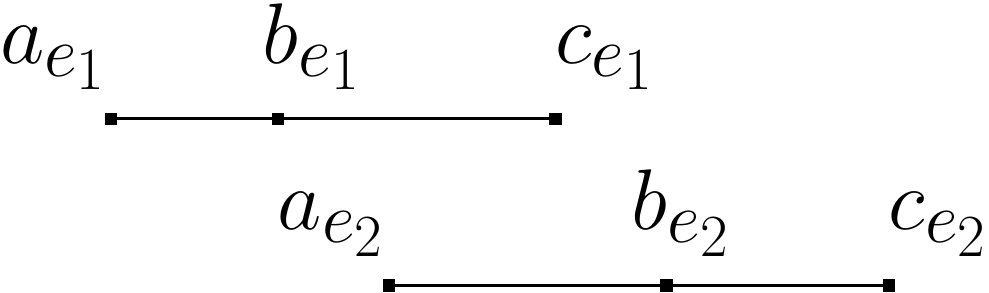}
		}
		\caption{Cases regarding the positioning of the exploration windows of any two edges.}
	\end{figure}

	Note that in case~\ref{item:a1}, it is impossible to explore $e_2$ using its first window, as that would disallow using any exploration window for $e_1$. Therefore, the value $x_{e_2} = 1$ is forced. Similarly, in case~\ref{item:a2}, exploring $e_2$ using its second window would disallow exploring $e_1$ at all. So, the value $x_{e_2}=0$ is forced. 
	Note also that in case~\ref{item:b}, setting $x_{e_1}=1$ would imply $x_{e_2}=1$, while setting $x_{e_1}=0$ allows for $x_{e_2}$ to be set to either $0$ or $1$. So, in the construction of the formula $F$ to be satisfied, we shall add the clause $(\neg{x_{e_1}} \vee x_{e_2})$. 
We proceed with constructing $F$.
	
	\begin{description}
		\item[Step 1.] We sort the $3n$ labels of $(G_s,L)$ in ascending order and we visit them from left to right. We say that an edge $e$ is ``open'' if we have passed through $a_e$ but not through $c_e$, in the increasing order of labels. Now we create two arrays $A$, $B$. 
Array $A$ shall hold those open edges $e$ , whose middle label $b_e$ we have not reached yet, and $B$ shall hold those open edges $e$, whose $b_e$ we have passed. 
				We start with the first label in the order, adding the corresponding edge to $A$. We move on to the next label. If a label we encounter is the first label, $a_e$, of some edge $e$, we merely add $e$ to $A$ and move on to the next label in the order. If a label we encounter is the second label, $b_e$, of some edge $e$, then we:
		\begin{enumerate}
			\item Remove $e$ from $A$.
			\item For every $e'\in A$, we set $x_{e'}=1$; for every $e''\in B$, we set $x_{e''}=0$.
			\item Add $e$ to $B$.
		\end{enumerate}
		Then, we move on to the next label. If a label we encounter is the third label, $c_e$, of some edge $e$, then we remove $e$ from $B$ and move on to the next label.
				If at any point in the above process we set some variable $x_e = 1$ (resp.~$x_e=0$) that was previously set to $0$ (resp.~$1$) then we stop and decide that $(G_s,L)$ is not explorable. Otherwise, we proceed to step 2.
				Clearly, the running time required for this step is dominated by the time needed to sort the labels: $O(n\log{n})$.
		\item[Step 2.] Step 1 sets the values of the $x$-variables of all edges $e$, whose second label lies within the time interval $I_{e'}$ of some other edge $e'$. At this point, we have some edges whose corresponding $x$-variables have been fixed (i.e.~forced) to some truth value (\emph{fixed} edges) and some edges whose corresponding $x$-variables have not been fixed (\emph{non-fixed} edges). 
		
		Within Step 2, we check whether any forced truth value (from Step 1) for a variable is conflicting with the forced truth value of another variable. To do so, we only consider two of the three labels for each edge $e$ whose variable $x_e$ has been forced within Step~1. Namely, if $x_e=0$ (resp.~$x_e=1$), we only consider the labels $a_e,b_e$ (resp.~$b_e,c_e$) of edge $e$ and we ignore label $c_e$ (resp.~$a_e$). 
		For each such edge $e$ (i.e.~for each edge whose truth value has been forced within Step 1) 
		these two ``considered'' labels form an interval (i.e.~the interval $I_e^1$ if $x_e=0$ and the interval $I_e^2$ if $x_e=1$. 
Now, in $O(n)$ time we can scan all these intervals (of the edges whose truth value has been forced within Step 1) and in the same time we can check whether any pair of them overlaps. If so, we stop  and decide that $(G_s,L)$ is not explorable. Otherwise, we proceed to step 3.

		\item[Step 3.] Within Step 3, we only deal with non-fixed edges. We create arrays $A$ and $B$ on these edges, as in Step 1. We visit the labels of the non-fixed edges in increasing order, i.e.~from left to right. We start with the first label, adding the corresponding edge to $A$, and move on to the next label. If a label we encounter is the second label, $b_e$, of some edge $e$, then we move $e$ from $A$ to $B$, and move on to the next label. If a label we encounter is the third label, $c_e$, of some edge $e$, then we:
		\begin{enumerate}
			\item Remove $e$ from $B$.
			\item Check whether there is an $e'$ currently in $A$. If so, we add to $F$ the clause $(\neg{x_e} \vee x_{e'})$. Notice that $A$ always contains at most one edge; otherwise the middle label of one of the edges currently in $A$ would lie within the exploration window of another edge in $A$, and thus one of them would be a fixed edge, contradiction.
			\item Move on to the next label.
		\end{enumerate}
		
		Notice that in the above procedure, for every edge $e$ we may add a clause containing $x_e$ or $\neg{x_e}$ to $F$ at most twice. Therefore, the total number of clauses in $F$ is~$O(n)$.
		
		\item[Step 4.] Within Step 4, we only deal with pairs of one fixed and one non-fixed edge.
For each fixed edge $e_i$, only consider its forced exploration window, i.e.~the window $I_{e_i}^1$ if $x_{e_i}=0$ and the window $I_{e_i}^2$ if $x_{e_i}=1$. 
Note that the forced exploration window of any fixed edge $e'$ cannot contain the middle label $b_e$ of any non-fixed edge $e$ (see Figure~\ref{fig:fixed_non_fixed_1}). 
Indeed, otherwise we would set both $x_{e'}=0$ and $x_{e'}=1$ within Step 1, and thus we would have already decided that $(G_s,L)$ is not explorable. 
		\begin{figure}[h]
			\centering
			\subfloat[The middle label $b_e$ of a non-fixed edge $e$ cannot lie within the forced exploration window of a fixed edge.\label{fig:fixed_non_fixed_1}]{%
				\includegraphics[width=0.26\linewidth]{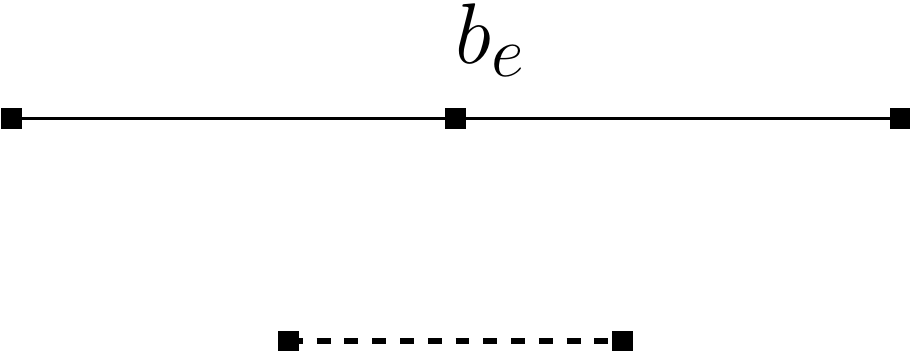}
			}
			\hfill
			\subfloat[The middle label $b_e$ of a non-fixed edge $e$ is to the ``left'' of the forced exploration window of a fixed edge.\label{fig:fixed_non_fixed_2}]{%
				\includegraphics[width=0.26\linewidth]{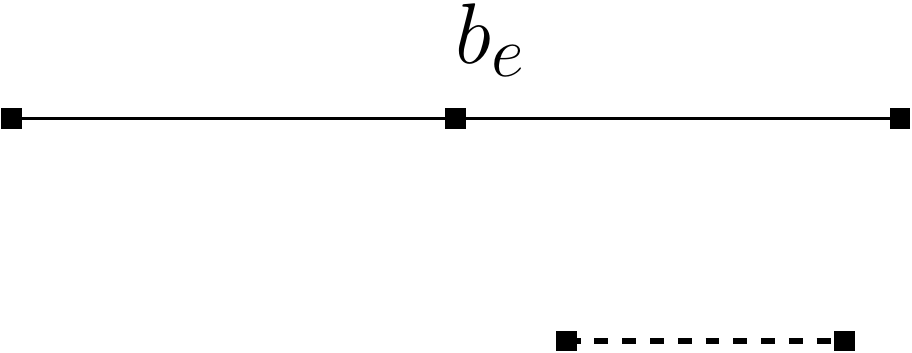}
			}
			\hfill
			\subfloat[The middle label $b_e$ of a non-fixed edge $e$ is to the ``right'' of the forced exploration window of a fixed edge.\label{fig:fixed_non_fixed_3}]{%
				\includegraphics[width=0.26\linewidth]{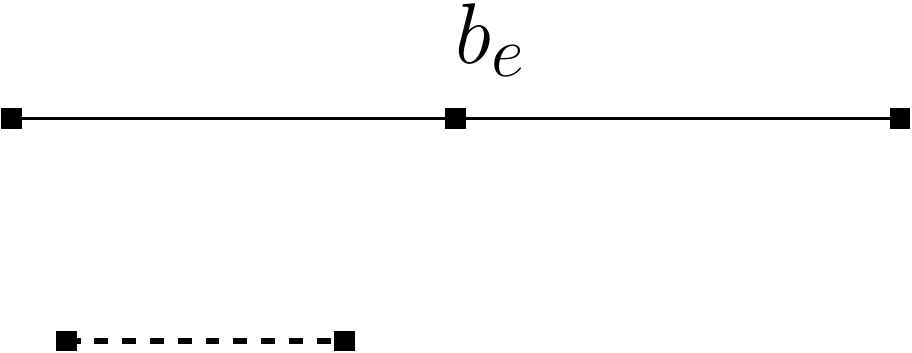}
			}
		\caption{Cases regarding the positioning of the exploration windows of pairs of one non-fixed edge and one fixed edge. In this figure, the exploration windows of non-fixed edges are drawn with solid lines, 
while the forced exploration windows of fixed edges are drawn with dashed lines.}
		\end{figure}
		
		So, it either is the case that $b_e$ is smaller than the smallest label of the forced window of the fixed edge (see Figure~\ref{fig:fixed_non_fixed_2}), or $b_e$ is larger than the largest label of the forced window of the fixed edge (see Figure~\ref{fig:fixed_non_fixed_3}). In the first case, we are forced to set $x_e=0$, while in the second case we are forced to set $x_e=1$.
		
		Now consider, in increasing order, all three labels for every non-fixed edge and the two labels of the forced window of every fixed edge. We scan through the sorted list of all these labels twice, as follows:
		
		\begin{enumerate}[(1)]
			\item In the first scan, we go through the labels to detect whether the label $a_e$ of a non-fixed edge $e$ lies within the forced exploration window of some fixed edge. Note here that there cannot be any two non-fixed edges $e_1,e_2$ whose labels $a_{e_1},a_{e_2}$ lie within the forced exploration window of the same fixed edge, since otherwise one of $e_1,e_2$ would have been classified as a fixed edge in Step 1, contradiction. 
Furthermore note that, if there is a non-fixed edge $e$ such that its label $a_e$ lies within the forced exploration window of a forced edge $e'$, then the middle label $b_e$ of $e$ lies to the right of the forced exploration window of $e'$; otherwise, a conflict would have been detected in Step 1.\\
			If the label $a_e$ of a non-fixed edge $e$ lies within the forced exploration window of a fixed edge, then we set $x_e=1$.
			\item In the second scan, we go through the same set of labels to detect whether the forced exploration window of a fixed edge starts within the interval $I_e$ of a non-fixed edge~$e$. 
			Let $e_1,e_2$ be two non-fixed edges and let $e'$ be a fixed edge. Note that the forced exploration window of $e'$ cannot start within the intersection $I_{e_1}\cap I_{e_2}$; indeed, otherwise either one of $e_1,e_2$ would then have been classified as a fixed edge in Step 1, or there would have been a conflict detected in Step 1.\\
			Now let $e$ be a non-fixed edge and $e'$ be a fixed edge. Note that, if the forced exploration window of $e'$ starts within the interval $I_e$, then the forced exploration window of $e'$ finishes \emph{after} the interval $I_e$ finishes; indeed, otherwise either a conflict would have been detected in Step 1 or the $e$ would have been classified as a fixed edge in Step 1.\\
			If the forced exploration window of a fixed edge $e'$ starts within the interval $I_e$ of a non-fixed edge $e$, then we check whether $x_e$ was set to $1$ in the first scan of the labels. If so, then we stop  and decide that $(G_s,L)$ is not explorable. Otherwise, we set $x_e=0$.
		\end{enumerate}
		Now, for every edge $e$ whose $x_e$ was set to $1$ we add to $F$ the clause $(x_e)$; for every edge $e$ whose $x_e$ was set to $0$ we add to $F$ the clause $(\neg{x_e})$. Notice that $F$ still contains $O(n)$ clauses in total.
		
		\item[Step 5.] We answer that $(G_s,L)$ is explorable if the formula $F$ is satisfiable and that $(G_s,L)$ is not explorable, otherwise.
	\end{description}
	
	The above procedure solves \textsc{StarExp(3)} in $O(n\log{n})$ time for instances $(G_s,L)$ with \emph{distinct} labels. 
\end{proof}

\begin{corollary}
\label{decision-3-labels-general-cor}
	\textsc{StarExp(3)} can be solved in $O(n\log{n})$ time on arbitrary instances.
\end{corollary}
\begin{proof}
	One can easily reduce the case of instances of \textsc{StarExp(3)}, where there exist $e,e' \in E$ with some $l \in L(e) \cap L(e')$, to the case of instances with distinct labels (which was dealt with in Theorem~\ref{thm:lognStarExp3}):	
	\begin{description}
		\item[Case 1.] If $c_e = a_{e'}$ then, equivalently, slightly move $c_e$ to the right.
		\item[Case 2.] If $c_e = b_{e'}$ then, equivalently, slightly move $b_{e'}$ to the right.	
		\item[Case 3.] If $c_e = c_{e'}$ then, equivalently, slightly move $c_e$ to the left.
		\item[Case 4.] If $a_e = a_{e'}$ then, equivalently, slightly move $a_e$ to the right.
		\item[Case 5.] If $a_e = b_{e'}$ then, equivalently, slightly move $a_{e}$ to the left.
		\item[Case 6.] If $b_e = b_{e'}$ then there is no exploration.
	\end{description}
It is easy to check that the above slight movements of labels do not affect the explorability of the input instance. Therefore, as all these modifications can be performed in $O(n)$ time, th statement of the corollary follows by Theorem~\ref{thm:lognStarExp3}.
\end{proof}

\begin{figure}[h]
	\centering
	\includegraphics[width=0.23\textwidth]{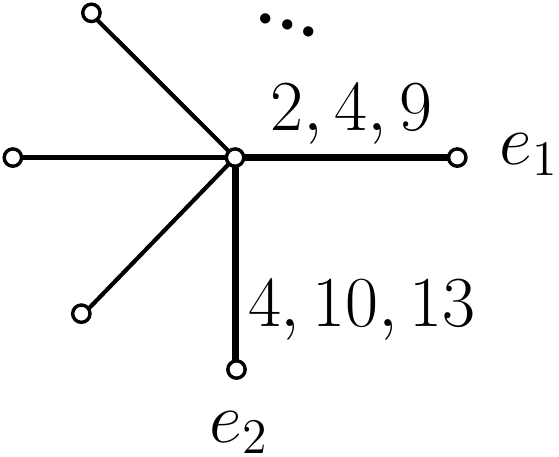}
	\caption{Edges $e_1$ and $e_2$ share label $4$. If we consider task $T_1$ associated with $e_1$ to have possible scheduling intervals $(2,4)$ and $(4,9)$, and task $T_2$ associated with $e_2$ to have possible scheduling intervals $(4,10)$ and $(10,13)$, then a feasible schedule would execute $T_1$ during $(2,4)$ and $T_2$ during $(4,10)$. However, this does not correspond to an exploration of the edges $e_1$ and $e_2$, since exiting $e_1$ with label $4$ would disallow entry to $e_2$ with the same label (due to our definition of journeys)! Therefore, we add a small positive number, e.g.~$\varepsilon = 0.5$, to the finish time of every possible scheduling interval for all tasks corresponding to the edges of the star graph and we can solve the problem by solving the resulting scheduling problem instance. Task $T_1$ would then have possible scheduling intervals $(2,4.5)$ and $(4,9.5)$, and task $T_2$ would have possible scheduling intervals $(4,10.5)$ and $(10,13.5)$.}
	\label{fig:label_conflict}
\end{figure}
\begin{observation*}
	\textsc{StarExp(3)} is similar to the single processor scheduling problem with discrete starting times when each task has at most two possible starting times~\cite{keil}. Every edge $e$ with labels $l_1<l_2<l_3$ can be viewed as a task to be scheduled, and the possible starting times for that task are $l_1$ and $l_2$. However, in our case the execution times vary depending on which starting time we choose for $e$; if we schedule it with starting time $l_1$, then the execution lasts $l_2-l_1+\varepsilon$ time, whereas if we schedule it with starting time $l_2$, then the execution lasts $l_3-l_2+\varepsilon$ time, for some constant $\varepsilon \in (0,1)$. The addition of $\varepsilon$ is done to ensure that two edges that have the same label are not both scheduled using that label (see Figure~\ref{fig:label_conflict} for an example).
\end{observation*}

%

\section{Hardness for $k\geq 6$ labels per edge}\label{sec:L-reduction}

In this section we show that, whenever where $k \geq 6$, \SE\ is NP-complete 
and \MSE\ is APX-hard. Thus, in particular, \MSE\ does not admit a Polynomial-Time Approximation Scheme (PTAS), unless P~=~NP. Furthermore, due to the polynomial-time constant-factor approximation algorithm for \MSE\, which we provide in Section~\ref{greedy-sec}, it follows that \MSE\ is also APX-complete. 
We prove our hardness results through a reduction from a special case of \textsc{3SAT}, namely \textsc{3SAT(3)}, defined below. This problem is known to be NP-complete~\cite{papadimitriou} and its maximization variant (i.e.~\textsc{MAX3SAT(3)}, where the aim is to maximize the number of clauses that can be simultaneously satisfied) is APX-complete~\cite{ausiello}.

\vspace{0cm} \noindent \fbox{ 
	\begin{minipage}{0.96\textwidth}
		\begin{tabular*}{\textwidth}{@{\extracolsep{\fill}}lr} 3SAT(3) \ & \\ \end{tabular*}
		
		\vspace{1.2mm}
		{\bf{Input:}} A boolean formula in CNF with variables $x_1,x_2, \ldots, x_p$ and clauses $c_1,c_2,\ldots,c_q$, such that each clause has at most $3$ literals, and each variable appears in at most $3$ clauses.\\
		{\bf{Output:}} Decision on whether the formula is satisfiable.
\end{minipage}} \vspace{0,3cm}

\paragraph*{Intuition and overview of the reduction}
Given an instance $F$ of 3SAT(3), we shall create an instance $(G_s,L)$ of \SE\ such that $F$ is satisfiable if and only if $(G_s,L)$ is explorable.
Henceforth, we denote by $|\tau(F)|$ the number of clauses of $F$ that are satisfied by a truth assignment $\tau$ of $F$.
Without loss of generality we make the following assumptions on $F$. Firstly, if a variable occurs only with positive (resp. only with negative) literals, then we trivially set it to $true$ (resp.~$false$) and remove the associated clauses. Furthermore, without loss of generality, if a variable $x_i$ appears three times in $F$, we assume that it appears once as a negative literal $\neg{x_i}$ and two times as a positive literal $x_i$; otherwise we rename the negation with a new variable. Similarly, if $x_i$ appears two times in $F$, then it appears once as a negative literal $\neg{x_i}$ and once as a positive literal $x_i$. 

Before we move on to the specifics of our reduction, we shall introduce the intuition behind it.
$(G_s,L)$ will have one edge corresponding to each clause of $F$, and three edges (one ``primary'' and two ``auxiliary'' edges) corresponding to each variable of $F$. 
We shall assign labels in pairs to those edges so that it is possible to explore an edge only by using labels from the same pair to enter and exit the edge; for example, if an edge $e$ is assigned the pairs of labels $l_1,l_2$ and $l_3,l_4$, with $l_1<l_2<l_3<l_4$, we shall ensure that one cannot enter $e$ with, say, label $l_2$ and exit with, say, label $l_3$. 
In particular, for the ``primary'' edge corresponding to a variable $x_i$ we will assign to it two pairs of labels, namely $(\alpha i - \beta, \alpha i - \beta + \gamma)$ and $(\alpha i + \beta, \alpha i + \beta + \gamma)$, for some $\alpha, \beta, \gamma \in \mathbb{N}$. The first (entry,exit) pair corresponds to setting $x_i$ to false, while the second pair corresponds to setting $x_i$ to true. We shall choose $\alpha, \beta, \gamma$ so that the entry and exit from the edge using the first pair is not conflicting with the entry and exit using the second pair. 

Then, to any edge corresponding to a clause $c_j$ that contains $x_i$ unnegated, we shall assign an (entry, exit) pair of labels $(\alpha i - \delta, \alpha i - \delta + \varepsilon)$, choosing $\delta, \varepsilon \in \mathbb{N}$ so that $(\alpha i - \delta, \alpha i - \delta + \varepsilon)$ is in conflict with the $(\alpha i - \beta, \alpha i - \beta + \gamma)$ pair of labels of the edge corresponding to $x_i$, which is associated with $x_i = false$ but not in conflict with the $(\alpha i + \beta, \alpha i + \beta + \gamma)$ pair. 
If $x_i$ is false in $F$ then $c_j$ cannot be satisfied through $x_i$ so we should not be able to explore a corresponding edge via a pair of labels associated with $x_i$. If $c_j$ contains $x_i$ negated, we shall assign to its corresponding edge an (entry, exit) pair of labels $(\alpha i + \zeta, \alpha i + \zeta + \theta)$, choosing $\zeta, \theta \in \mathbb{N}$ so that the latter is in conflict with the $(\alpha i + \beta, \alpha i+- \beta + \gamma)$ pair of labels of the edge corresponding to $x_i$, which is associated with $x_i = true$ but not in conflict with the $(\alpha i - \beta, \alpha i - \beta + \gamma)$ pair. If $x_i$ is true in $F$ then $c_j$ cannot be satisfied through $\neg{x_i}$ so we should not be able to explore a corresponding edge via a pair of labels associated with $\neg{x_i}$. 

Finally, for every variable $x_i$ we also introduce two additional ``auxiliary'' edges: the first one will be assigned the pair of labels $(\alpha i, \alpha i + \xi)$, $\xi \in \mathbb{N}$, so that it is not conflicting with any of the above pairs -- the reason for introducing this first auxiliary edge is to avoid entering and exiting an edge corresponding to some variable $x_i$ using labels from different pairs. 
The second auxiliary edge for variable $x_i$ will be assigned the pair of labels $(\alpha i + \chi, \alpha i + \chi + \psi)$, $\chi, \psi \in \mathbb{N}$, 
so that it is not conflicting with any of the above pairs -- the reason for introducing this edge is to avoid entering an edge that corresponds to some clause $c_j$ using a label associated with some variable $x_i$ and exiting using a label associated with a \emph{different} variable~$x_{i'}$.

\paragraph*{The reduction}
For the reduction from 3SAT(3) to \SE\, we select constants $\alpha, \beta, \gamma, \delta, \varepsilon, \zeta, \theta, \xi, \chi, \psi$ so that all the requirements mentioned above regarding the conflicts of different pairs of labels are satisfied.
In particular, given an instance $F$ of 3SAT(3), we create in polynomial time the following instance $(G_s,L)$ of \SE:
\begin{itemize}
	\item For every variable $x_i,~i=1,2,\ldots,p$, create an edge $e_i$ (the ``primary'' edge for $x_i$) with labels $50i-10$, $50i-7$, $50i+10$, and $50i+13$. The pair of labels $50i-10$, $50i-7$ of $e_i$ represent the assignment $x_i=false$ and the pair of labels $50i+10$, $50i+13$ represent the assignment $x_i=true$. More specifically, entry in $e_i$ with label $50i-10$ and exit from $e_i$ with label $50i-7$ represents $x_i=false$ (and similarly for the other pair of labels and the assignment $x_i=true$).
	\item For every variable $x_i,~i=1,2,\ldots,p$, also create an edge $e_i'$ (the first ``auxiliary'' edge for $x_i$) with labels $50i$ and $50i+1$. The labels on $e_i'$ ensure that we do not enter $e_i$ with a label associated with the assignment $x_i=false$ and exit from $e_i$ with a label associated with the assignment $x_i=true$; in an exploration of $(G_s,L)$, this will not occur since we must explore both $e_i'$ and $e_i$, and this only happens if we enter and exit $e_i$ using labels associated with the same truth assignment for the variable $x_i$.
	\item For every variable $x_i,~i=1,2,\ldots,p$, also create an edge $e_i''$ (the second ``auxiliary'' edge for $x_i$) with labels $50i+15$ and $50i+16$. The labels on $e_i''$ ensure that we do not enter and exit edges associated with clauses of $F$ (see bullet point below) using pairs of labels that are associated with different variables.
	\item For every clause $c_j,~j=1,2,\ldots,q$, create an edge $e_{p+j}$ with the following labels:
	\begin{itemize}
		\item For every variable $x_i$ that appears unnegated for the \emph{first}\footnote{We consider here the order $c_1,c_2,\ldots, c_q$ of the clauses of $C$; we say that $x_i$ appears unnegated for the \emph{first} time in some clause $c_\mu$ if $x_i \not\in c_m,~m<\mu$.} time in $C$, add two labels $50i-12$ and $50i-9$.
		\item For every variable $x_i$ that appears unnegated for the \emph{second}\footnote{Again, we consider the order $c_1,c_2,\ldots, c_q$ of the clauses of $C$; we say that $x_i$ appears unnegated for the \emph{second} time in some clause $c_\mu$ if $\exists m<\mu$, such that $x_i \in c_m$.} time in $C$, add two labels $50i-8$ and $50i-5$. Note here that both (entry,exit) pairs $(50i-12,50i-9)$ and $(50i-8,50i-5)$ are conflicting with the (entry,exit) pair $(50i-10,50i-7)$ of the edge $e_i$ that is associated with the assignment $x_i=false$.
		\item For every variable $x_i$ that appears negated, add two labels $50i+8$ and $50i+11$. Note that the (entry,exit) pair $(50i+8,50i+11)$ is conflicting with the (entry,exit) pair $(50i+10,50i+13)$ of the edge $e_i$ that is associated with the assignment $x_i=true$.
	\end{itemize}
\end{itemize}
The reader is referred to Figure~\ref{fig:poly_reduction} for an example construction.

\begin{figure}[h]
	\centering
     \subfloat[]{%
		\includegraphics[width=0.35\linewidth]{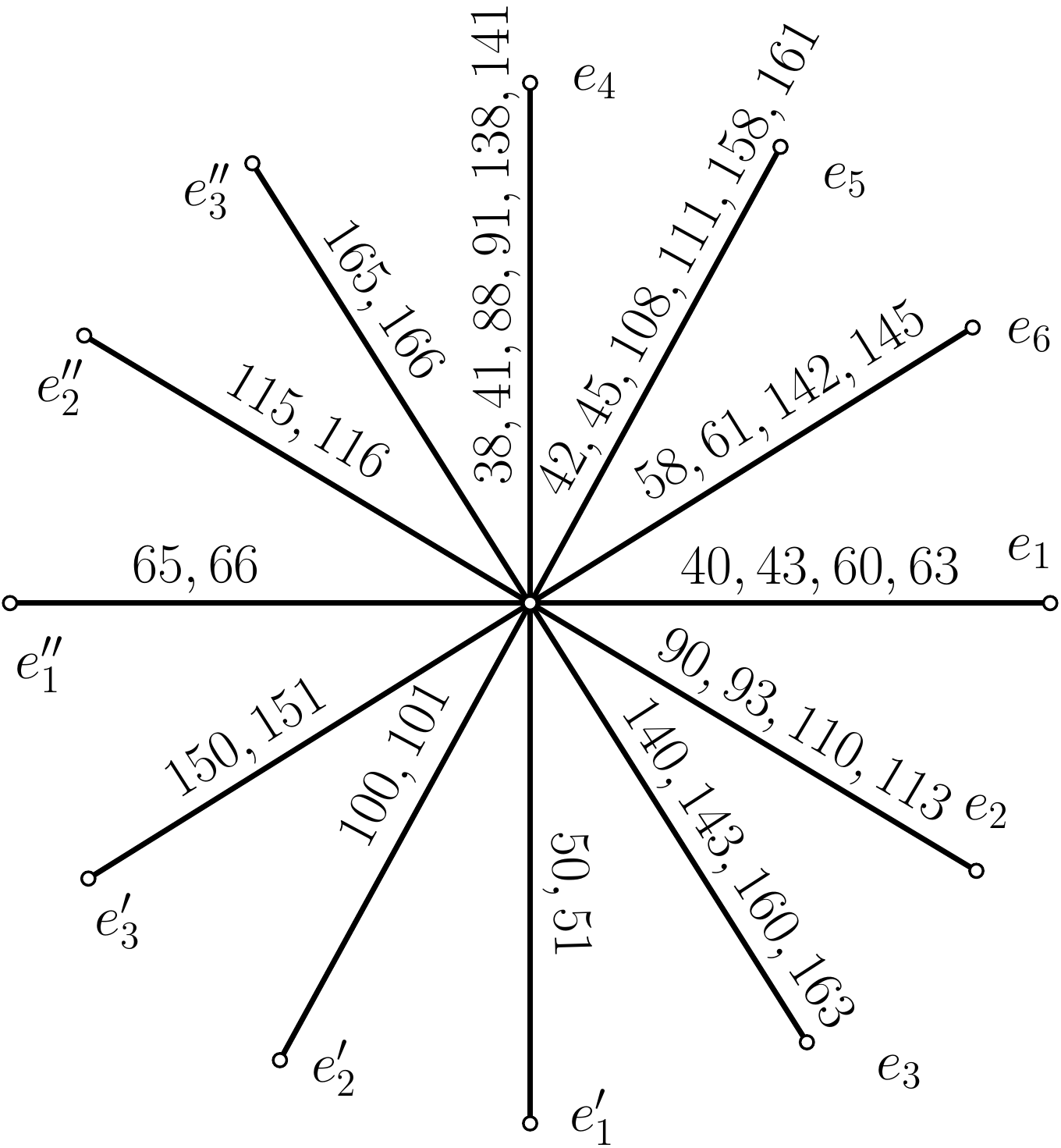}
	}
	\hfill
	\subfloat[]{%
		\includegraphics[width=0.35\linewidth]{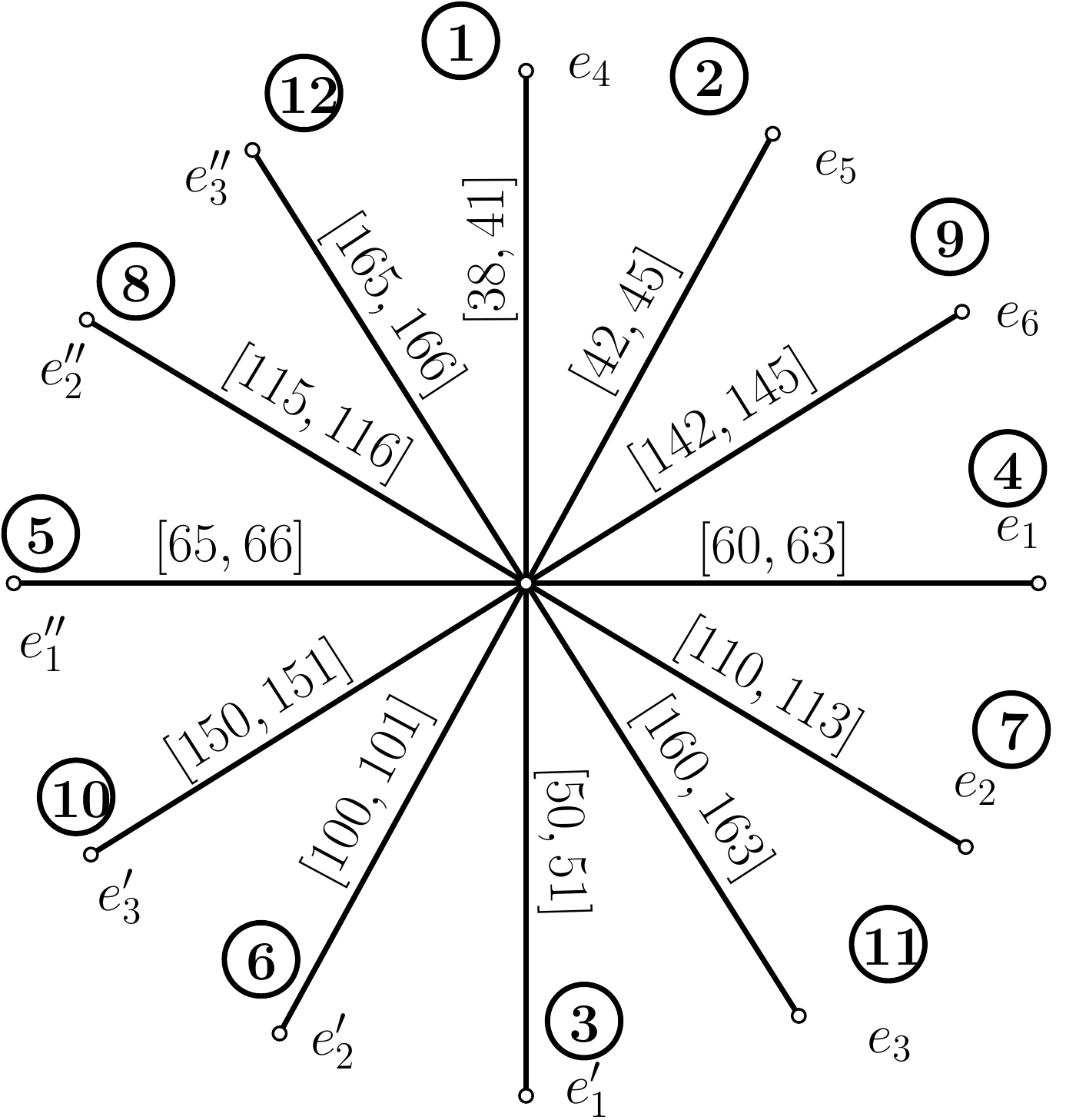}
	}
	\caption{The temporal star constructed for the formula $(x_1 \vee x_2 \vee x_3 ) \wedge (x_1 \vee \neg{x_2} \vee \neg{x_3}) \wedge (\neg{x_1} \vee x_3)$. Setting $x_1$ to true, $x_2$ to true and $x_3$ to true yields a satisfying truth assignment whose corresponding exploration is indicated in (b), where the numbers in the circles indicate the order over time of the exploration of each edge.}\label{fig:poly_reduction}
\end{figure}

Notice that the (entry,exit) pairs on edges associated with different clauses are not conflicting if we pair them in ascending order, i.e.~if an edge has labels $l_1,l_2,l_3,l_4$ we pair them into $(l_1,l_2)$, $(l_3,l_4)$; that is because each variable appears at most twice unnegated and once negated. 
The following lemmas are needed for the proofs of NP-completeness (Theorem~\ref{thm:NP-comp}) and APX-hardness (Theorems~\ref{thm:L-reduc} and~\ref{thm:epsilon-bound}).

\begin{lemma}\label{lem:max_size_explore}
	There exists a (partial) exploration $J$ of $(G_s,L)$ of maximum size which explores all the edges $e_i,e_i',e_i''$, $i=1,2,\ldots,p$.
\end{lemma}
\begin{proof}
	Let $J$ be a (partial) exploration of $(G_s,L)$ of maximum size. Without loss of generality, we may assume that $J$ explores every edge at most once.
	
	We will show that one may also assume that any edge of the form $e_{p+j}$ corresponding to the clause $c_j$, $j=1,\ldots, q$, that is explored by $J$ is explored via (entry,exit) pairs associated with the same literal of $c_j$. Indeed, let $c_j = (l_\pi \vee l_\rho \vee l_\sigma)$ be a clause whose corresponding edge is visited by $J$, where the literal $l_\pi$ (resp. $l_\rho$ and $l_\sigma$) is $x_\pi$ (resp. $x_\rho$ and $x_\sigma$) or $\neg{x_\pi}$ (resp. $\neg{x_\rho}$ and $\neg{x_\sigma}$), for some $\pi = 1,\ldots, p$ (resp. $\rho = 1,\ldots, p$ and $\sigma = 1,\ldots, p$).
	
	Let $\alpha_1, \ldots, a_6$ the labels of $e_{p+j}$. If $J$ explores $e_{p+j}$ using $(\alpha_2, \alpha_3)$ (resp. $(\alpha_4, \alpha_5)$), then there exists an edge $e_r''$, for some $r=1, \ldots, p$, which is not explored by $J$, by construction of the edges of the form $e_i''$. Then, one can create a (partial) exploration of the same size as $J$ as follows:	starting from $J$ swap the exploration of $e_{p+j}$ with the exploration of $e_r''$. Note that there is no other edge that is explored by $J$ using a time window that overlaps with that of the exploration of $e_r''$, since it would also be conflicting with the window used by $J$ to explore $e_{p+j}$. By iteratively swapping edges of the form $e_{p+j}$ - that are explored using (entry, exit) pairs associated with different literals - with edges of the form $e_r''$, we result in a (partial) exploration $J'$ of the same size as $J$ which explores all edges of the form $e_{p+j}$ via (entry,exit) pairs associated with the same literal of $c_j$. Note that no exploration window of any edge $e_i',e_i''$ overlaps with any exploration window of any edge $e_{p+j}$ in $J'$. In fact, since $J'$ is of maximum size, all edges of the form $e_i''$ must be explored by $J'$.
	
	We will now show that we may also assume that all edges of the form $e_i'$ are explored by a maximum size (partial) exploration of $(G_s,L)$. Assume that an edge $e_i'$ is not explored by $J'$, $i=1, \ldots,p$. Then it must be that the edge $e_i$ corresponding to the variable $x_i$ is explored by $J'$ using the pair $(50i-7, 50i+10)$ (as this is the only possible conflicting exploration window). Construct an exploration $J''$ as follows: starting from $J'$,
	\begin{enumerate}
		\item \label{item:replace} replace $(50i-7, 50i+10)$ by $(50i+10, 50i+13)$ as the exploration window of $e_i$, and
		\item add $e_i'$ to the exploration using its window $(50i, 50i+1)$.
	\end{enumerate}
	
	Notice that step~\ref{item:replace} is indeed possible without causing conflicts with other edges: let $c_\alpha$ be the clause containing $\neg{x_i}$ (so, the edge $e_{p+\alpha}$ has been assigned, amongst others, the labels $50i+8$, $50i+11$); then, since $J'$ explores $e_i$ using $(50i-7, 50i+10)$ it must be that $e_{p+\alpha}$ is not explored by $J'$ -if explored at all- using $(50i+8, 50i+11)$. So, swapping the windows as shown in step~\ref{item:replace} is possible without conflicts. Now, notice that $J''$ has size larger than the size of $J'$ which is a contradiction. Therefore, $e_i$ cannot be explored by $J'$ (or any maximum size exploration of $(G_s,L)$) using $(50i-7, 50i+10)$, and $e_i'$ must be explored by $J'$, for all $i=1, \ldots,p$.
	
	It remains to show that all edges $e_i,~i=1, \ldots,p$, are explored by $J'$. Assume that there is an edge $e_i$, for some $i=1,\ldots, p$, that is not explored by $J'$ and let $c_\alpha$ be the clause that contains $\neg{x_i}$. The only way that $e_i$ cannot be explored by $J'$ is if $J'$ explores edges that cause a conflict with both exploration windows $(50i-10, 50i-7)$ and $(50i+10, 50i+13)$ of $e_i$. In fact, if $e_i$ is not explored by $J'$ then it must be that $e_{p+\alpha}$ is explored using the exploration window $(50i+8, 50i+11)$. Then we can create $J''$ of same size as $J'$, starting from $J'$, by removing $e_{p+\alpha}$ from the exploration and adding $e_i$ to the exploration, exploring it using the window $(50i+10, 50i+13)$. This way, one can create a maximum size exploration that contains all edges $e_i$, $i=1,\ldots, p$.
	
	We conclude that there can always be found an exploration of maximum size which explores all edges $e_i,e_i',e_i'',~i=1,\ldots,p$, which completes the proof of the lemma.
\end{proof}

\begin{lemma}\label{lem:apx}
	There exists a truth assignment $\tau$ of $F$ with $|\tau(F)| \geq \beta$ if and only if there exists a (partial) exploration $J$ of $(G_s,L)$ of size $|J|\geq 3p + \beta$.
\end{lemma}
\begin{proof}
	($\Rightarrow $) Assume that there is a truth assignment $\tau$ that satisfies $\beta$ clauses of $F$. We give a (partial) exploration $J$ of $(G_s,L)$, of size $3p+\beta$, as follows. First, we add to $J$ all the edges $e_i',e_i'',~i=1,2,\ldots, p$; these are $2p$ edges in total and can only be explored one way as they have each been assigned two labels. Then, we add to $J$ all edges $e_i,~i=1,2,\ldots,p$ which are $p$ edges in total; we explore each $e_i$ depending on the value of $x_i$ in $\tau$, namely if $x_i=true$ we explore $e_i$ using the pair $(50i+10,50i+13)$, and if $x_i=false$ we explore $e_i$ using the pair $(50i-10,50i-7)$. Now, consider an arbitrary clause $c_j$ of $F$ that is satisfied in $\tau$, i.e.~it has at least one true literal which is of the form $x_i$ or $\neg{x_i}$, for some $i=1,2,\ldots,p$. If $x_i =true$ then we explore $e_{p+j}$ using the pair of labels that corresponds to the unnegated appearance of $x_i$ in $c_j$ -- depending on whether $x_i$ appears unnegated for the first or the second time in $c_j$, this pair is $(50i-12,50i-9)$ or $(50i-8,50i-5)$, respectively. If $\neg{x_i} =true$ then we explore $e_{p+j}$ using the pair of labels $(50i+8,50i+11)$ that corresponds to the negated appearance of $x_i$ in $c_j$. As there are at least $\beta$ satisfied clauses of $F$ in $\tau$, we have added at least $\beta$ extra edges to $J$. Notice that it has already been established in the previous section that all the (entry,exit) pairs chosen for the exploration of the edges that we added in $J$ are pairwise non-conflicting. So, $J$ is a (partial) exploration of $(G_s,L)$ which explores at least $3p+\beta$ edges.
	
	($\Leftarrow $) Assume that there is a (partial) exploration of $(G_s,L)$ which explores at least $3p+\beta$ edges. By Lemma~\ref{lem:max_size_explore}, there is a (partial) exploration $J$ of $(G_s,L)$ of maximum size which explores all edges $e_i,e_i',e_i''$, for $i=1,2,\ldots,p$. It is $|J| \geq 3p + \beta$ and we know that $J$ already explores $3p$ edges that are not associated with clauses of $F$. So, it must be that $J$ also explores at least $\beta$ edges that are associated with clauses of $F$.
	We can construct a truth assignment $\tau$ of $F$ which satisfies at least $\beta$ clauses as follows. We check the (entry,exit) pairs of exploration in $J$ of the edges that correspond to clauses of $F$. The entry and exit labels must be associated with the same variable $x_i$, otherwise there would be conflict with the exploration of the respective edge $e_i''$. So, for each such edge, we set the variable $x_i$ associated to the chosen (entry,exit) pair to true if the pair corresponds to an unnegated appearance (for the first or second time in $F$) of the variable, and we set $x_i$ to false otherwise. Without loss of generality, we can set any remaining variables to true and it is easy to see that the resulting truth assignment satisfies at least $\beta$ clauses of $F$, each one associated with an edge $e_{p+j}$, $j=1,2,\ldots,q$, that is explored in $J$.
\end{proof}

We move on to the main theorem of the section.

\begin{theorem}\label{thm:NP-comp}
	\textsc{StarExp($k$)} is NP-complete for every $k\geq6$.
\end{theorem}

\begin{proof}
	It is easy to see that \SE\ is in NP, for every $k$. We may verify any solution, i.e.~exploration, in polynomial time by checking that it visits all $O(n)$ vertices and it enters and exits all $O(n)$ edges on existing edge-labels (we would need to check at most $k(n-1)$ labels in total). 
	
	An immediate corollary of Lemma~\ref{lem:apx} is that $F$ is satisfiable if and only if $(G_s,L)$ is explorable. Therefore, since the constructed instance $(G_s,L)$ from the reduction has at most~6 labels per edge, it follows that \textsc{StarExp(6)} is NP-complete.

	To extend this result to the NP-completeness of \SE\ also for values $k\geq 6$, it suffices to add to  the constructed instance $(G_s,L)$ an ``artificial'' edge $e^*$ that only contains labels that are much larger than any of the labels on the other edges of $(G_s,L)$. If $F$ is satisfiable then the exploration of $(G_s,L)$ is as described previously, with the addition of exploring $e^*$ using any of its exploration windows. This is possible, since none of those will be conflicting with any window of any other edge. Conversely, if $(G_s,L)$ is explorable, then the exploration of $e^*$ can be ignored regarding the satisfiability of $F$ since it overlaps with no other edge's exploration.
\end{proof}

Using Lemma~\ref{lem:apx} we can now prove the APX-hardness of \MSE.

\begin{theorem}\label{thm:L-reduc}
	\MSE\ is APX-hard, for $k\geq 6$.
\end{theorem}
\begin{proof}
	Denote by $OPT_{\textsc{Max3SAT(3)}}(F)$ the greatest number of clauses that can be simultaneously satisfied by a truth assignment of $F$.
	The proof is done by an \emph{L-reduction}~\cite{papadimitriou91}
	from the \textsc{Max3SAT(3)} problem, i.e.~by an approximation preserving
	reduction which linearly preserves approximability features. For such a
	reduction, it suffices to provide a polynomial-time computable function $g$
	and two constants $\gamma ,\delta >0$ such that:
	
	\begin{itemize}
		\item $OPT_{\textsc{MaxStarExp}}((G_s,L)) \leq \gamma \cdot OPT_{\textsc{Max3SAT(3)}}(F)$, for any boolean formula~$F$, and
		\item for any (partial) exploration $J'$ of $(G_s,L)$, $g(J')$ is a truth assignment for $F$ and $OPT_{\textsc{Max3SAT(3)}}(F)-|g(J')| \leq \delta ( OPT_{\textsc{MaxStarExp}}((G_s,L)) -|J|)$, where $|g(J')|$ is the number of clauses of $F$ that are satisfied by $g(J')$.
	\end{itemize}
	
	We will prove the first condition for $\gamma =19$. Note that a random truth assignment satisfies each clause of $F$ with probability at least $\frac{1}{2}$ (if each clause had exactly 3 literals, then it would be satisfied with probability $\frac{7}{8} $, but we have to account also for single-literal and two-literal clauses), and thus there exists an assignment $\tau $ that satisfies at least $\frac{q}{2}$ clauses of $F $. 
Furthermore, since every clause has at most 3 literals, it follows that $q\leq \frac{p}{3}$. 
Therefore $OPT_{\textsc{Max3SAT(3)}}(F)\geq \frac{q}{2} \leq \frac{p}{6}$, and thus $p\leq 6 \cdot OPT_{\textsc{Max3SAT(3)}}(F )$. Now Lemma~\ref{lem:apx} implies that:
	
	\begin{eqnarray}
	OPT_{\textsc{MaxStarExp}}((G_s,L)) &=& 3p+ OPT_{Max3SAT(3)}(F ) 
	\notag \\
	&\leq & 3\cdot 6 \cdot OPT_{Max3SAT(3)}(F ) + OPT_{Max3SAT(3)}(F )
	\notag \\
	&= & 19 \cdot OPT_{\textsc{Max3SAT(3)}}(F ) \notag
	\end{eqnarray}
	
	To prove the second condition for $\delta =1$, consider an arbitrary partial exploration $J'$ of $G_{s} (L)$. As described in the ($\Leftarrow $)-part of the proof of Lemma~\ref{lem:apx}, we construct in polynomial time a truth assignment $g(J')=\tau $ that satisfies at least $OPT_{\textsc{MaxStarExp}}((G_s,L)) -3p$ clauses of $F $, i.e. $|g(J')|=|\tau (F )|\geq |J'|-3p$. Then:
	\begin{eqnarray}
	OPT_{\textsc{Max3SAT(3)}}(F)-|g(J')| &\leq &  OPT_{\textsc{Max3SAT(3)}}(F ) - |J'| +3p   \notag \\
	&=&  OPT_{\textsc{MaxStarExp}}((G_s,L)) - 3p - |J'| + 3p \notag\\
	&=&  OPT_{\textsc{MaxStarExp}}((G_s,L)) - |J'| \notag
	\end{eqnarray}
	
	This completes the proof of the theorem. 
\end{proof}

Now we prove a correlation between the inapproximability bounds for the \MSE\ problem and \textsc{Max3SAT(3)}, as a result of the L-reduction presented in Theorem~\ref{thm:L-reduc}. 
Note that, since \textsc{Max3SAT(3)} is APX-hard~\cite{ausiello}, there exists a constant $\varepsilon_0 >0$ such that there exists no polynomial-time constant-factor approximation algorithm for \textsc{Max3SAT(3)} with approximation ratio greater than $(1-\varepsilon_0)$, unless P~=~NP.

\begin{theorem}
\label{thm:epsilon-bound}
	Let $\varepsilon_0>0$ be the constant such that, unless P~=~NP, there exists no polynomial-time constant-factor approximation algorithm for \textsc{Max3SAT(3)} with approximation ratio greater than $(1-\varepsilon_0)$. 
Then, unless P~=~NP, there exists no polynomial-time constant-factor approximation algorithm for \MSE\  with approximation ratio greater than $(1-\frac{\varepsilon_0}{19})$.
\end{theorem}
\begin{proof}

Let $\varepsilon>0$ be a constant such that there exists a polynomial-time approximation algorithm 
$\mathcal{A}$ for \MSE\ with ratio $(1-\varepsilon)$. Let $F$ be an instance of \textsc{MAX3SAT(3)} 
with $p$ variables and $q$ clauses. 
We construct the instance $(G_s,L)$ of \MSE\ corresponding to $F$, as described in the L-reduction 
(see Theorem~\ref{thm:L-reduc}). Then we apply the approximation algorithm $\mathcal{A}$ to $(G_s,L)$, 
which returns a (partial) exploration $J$. 
Note that $|J|\geq (1-\varepsilon) \cdot OPT_{\textsc{MaxStarExp}}$.
As described in the proof of Lemma~\ref{lem:apx}, we construct from $J$ in polynomial time 
a truth assignment $\tau$; we denote by $|\tau|$ the number of clauses in $F$ that are satisfied by the truth assignment $\tau$. 
It now follows from the proof of Theorem~\ref{thm:L-reduc} that:
$$
OPT_{\textsc{Max3SAT(3)}}(F)-|\tau| \leq   OPT_{\textsc{MaxStarExp}}((G_s,L)) - |J| 
\leq   19\varepsilon \cdot OPT_{\textsc{Max3SAT(3)}}(F) 
$$
Therefore $|\tau|\geq (1-19\varepsilon)\cdot OPT_{\textsc{Max3SAT(3)}}(F)$. 
That is, using algorithm $\mathcal{A}$, we can devise a polynomial-time algorithm for \textsc{MAX3SAT(3)} 
with approximation ratio $(1-19\varepsilon)$. Therefore, due to the assumptions of the theorem 
it follows that $\varepsilon \geq \frac{\varepsilon_0}{19}$, unless P~=~NP. 
This completes the proof of the theorem.
\end{proof}

\section{An efficient 2-approximation greedy algorithm for \MSE}\label{greedy-sec}

Note that instances of \MSE\ have the following property: An edge $e_i$ that has an exploration window with the earliest exit time is selected in an optimal solution. Indeed, suppose that it is not the case and consider an optimal solution not exploring $e_i$. One can exchange the explored edge of this solution that has earliest exit time with the edge $e_i$, exploring $e_i$ using its first exploration window.

The above property leads to the construction of a polynomial time greedy approximation for \MSE, which as shown below, achieves approximation ratio $2$.

\begin{algorithm}[H]	
	\KwData{a temporal star graph $(G_s,L)$ with at most $k$ labels per edge, $k \in \mathbb{N}^*$}
	\KwResult{a (partial) exploration of $(G_s,L)$}
	
	\renewcommand{\thealgocf}{}
	\SetAlgorithmName{Algorithm~GREEDY}\\
	
	Initialize the set of candidate edges to be $\mathcal{C} = E$\;
	Initialize the set of explored edges to be $Exp =\emptyset$\;
	t:=0\;
	\While{$\mathcal{C}\not=\emptyset$ \textbf{\em or} no $e \in \mathcal{C}$ has $2$ labels greater or equal to $t$}{
		Find $e\in \mathcal{C}$ to be explored with entry time at least $t$ and minimum exit time. Let $t_0$ be said exit time\;
		Add $e$ to the set of explored edges, $Exp$ (with exploration window from $t$ until $t_0$)\;
		Remove $e$ from the set of candidate edges, $\mathcal{C}$\;
		$t=t_0+1$\;
	}

	\caption{A 2-approximation algorithm for \MSE}
	\label{alg:greedy}
\end{algorithm}

\begin{theorem}
	Algorithm~GREEDY is a 2-approximation algorithm for \MSE, running in time $\Theta(kn^2)$. The approximation ratio of the algorithm is tight, i.e.~there is an instance of the problem where Algorithm~GREEDY achieves exactly ratio $2$ (see~Figure~\ref{fig:tight_approx}).
\end{theorem}
\begin{proof}
	Note that Algorithm~GREEDY is indeed polynomial, running in time $\Theta(kn^2)$, as there will be at most $|E|$ iterations of the while-loop, each of which requires $O(k|E|)$ time; one needs to examine all edges' labels in the worst case to determine whether an edge exists in $\mathcal{C}$ that can be explored starting at time $t$ or greater.
	
	Now, consider an instance $(G_s,L)$ of \MSE. Let $OPT((G_s,L))$ denote an optimal solution to the problem, with $|OPT((G_s,L))|$ being the maximum number of edges of $(G_s,L)$ that can be explored. Let $A((G_s,L))$ denote the exploration that Algorithm~GREEDY returns with $|A((G_s,L))|$ being the number of edges explored in that exploration. We will show that $|OPT((G_s,L))| \leq 2 |A((G_s,L))|$.
	
	Let $G_1(L_1)$ be the edges selected for exploration in $A((G_s,L))$ together with their labels (as assigned to them by $L$). Let $G_2(L_2)$ be the remaining edges together with their labels. $G_1(L_1)$ and $G_2(L_2)$ are both instances of \MSE. We will show that it holds that:
	\begin{equation}\label{eq:opt_greedy}
	|OPT((G_s,L))| \leq |OPT(G_1(L_1))| + |OPT(G_2(L_2))|
	\end{equation}
	Obviously, $|OPT(G_1(L_1))| \leq |A((G_s,L))|$; in fact, $|OPT(G_1(L_1))|=|A((G_s,L))|$ since all edges in $G_1(L_1)$ can be explored with it being derived by the exploration produced by Algorithm~GREEDY. So, equation~\eqref{eq:opt_greedy} becomes:
	\begin{equation}\label{eq:opt_greedy_2}
	|OPT((G_s,L))| \leq |A((G_s,L))| + |OPT(G_2(L_2))|
	\end{equation}
	
	Now, let $l_1<l_2<\ldots < l_{|A((G_s,L))|}$ be the exit times of the exploration windows (of the edges) selected by Algorithm~GREEDY, in ascending order. Let also $l_0 = 0$. It is the case that all exploration windows of all edges in $G_2(L_2))$ that have entry time in the time interval $(l_{j-1}, l_j)$ have finish time equal to $l_j$ or greater than $l_j$, for $j=1,2,\ldots, |A((G_s,L))|$; if not, then they would have been selected by Algorithm~GREEDY. Therefore, at most one of the windows of $G_2(L_2)$ with start time in the time interval $(l_{j-1}, l_j)$ can be in an optimal solution of $G_2(L_2)$. Since there are $|A((G_s,L))|$ such time intervals, we get $|OPT(G_2(L_2))| \leq |A((G_s,L))|$. So, equation~\eqref{eq:opt_greedy_2} becomes:
	\begin{equation*}
	|OPT((G_s,L))| \leq 2 |A((G_s,L))|
	\end{equation*}
	which completes the proof of the first part of the theorem.
	
	To show that the approximation given by Algorithm~GREEDY is tight, consider the instance of \MSE\ shown in Figure~\ref{fig:tight_approx}. Note that the possible exploration windows for each edge are shown below the edge. 
	\begin{figure}
		\centering
		\includegraphics[width=0.4\textwidth]{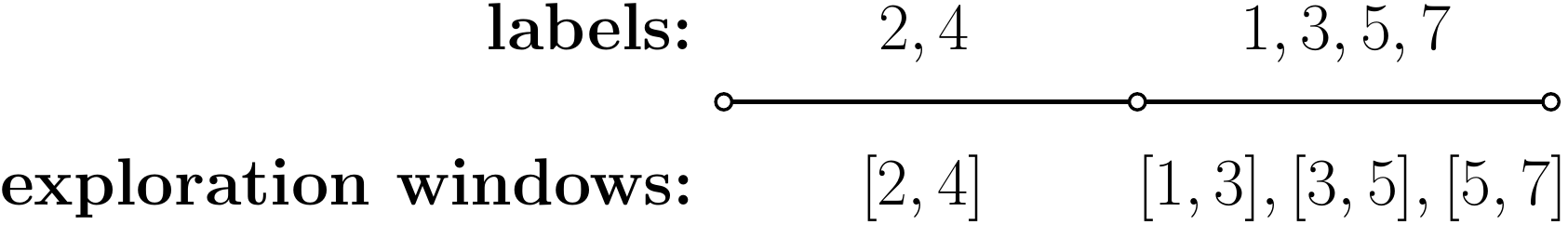}
		\caption{An instance of \MSE where the greedy approximation ratio is tight.}
		\label{fig:tight_approx}
	\end{figure}	
	
	It is easy to see that for this instance $(G_s,L)$, it holds that $|OPT((G_s,L))| = 2$, while $|A((G_s,L))| = 1$.
\end{proof}

\section{$k$ random models per edge}\label{sec:random}

We now study the problem of star exploration in a temporal star graph on an underlying star graph $G_s$ of $n$ vertices, where the labels are assigned to the edges of $G_s$ at random. In particular, each edge of $G_s$ receives $k$ labels independently of other edges, and each label is chosen uniformly at random and independently of others from the set of integers $\{1,2,\ldots, \alpha\}$, for some $\alpha \in \mathbb{N}$. We call this a \emph{uniform random temporal star} and denote it by $G_s(\alpha, k)$. In this section, we investigate the probability of exploring all edges in a uniform random temporal star based on different values of $\alpha$ and $k$, thus partially characterizing uniform random temporal stars that can be \emph{fully} explored or not, asymptotically almost surely.

\begin{theorem}
	If $\alpha \geq 2n$ and $k\geq 6n \ln{n}$, then the probability that we can explore all edges of $G_s(\alpha,k)$ tends to $1$ as $n$ tends to infinity.
\end{theorem}
\begin{proof}
	We consider the time-line from $1$ to $\alpha$ and we split it into $2n$ consecutive equal-sized time-windows of size $\frac{\alpha}{2n}$ as shown in Figure~\ref{fig:boxes}. Let us henceforth refer to those as \emph{boxes} and denote the $i^{th}$ such box by $B_i$. The first box contains the labels $1,2\ldots, \frac{\alpha}{2n}$, the second box contains the labels $\frac{\alpha}{2n}+1,\frac{\alpha}{2n}+2\ldots, \frac{\alpha}{n}$, and so on.
	
	\begin{figure}[h]
		\centering
		\includegraphics[width=0.7\textwidth]{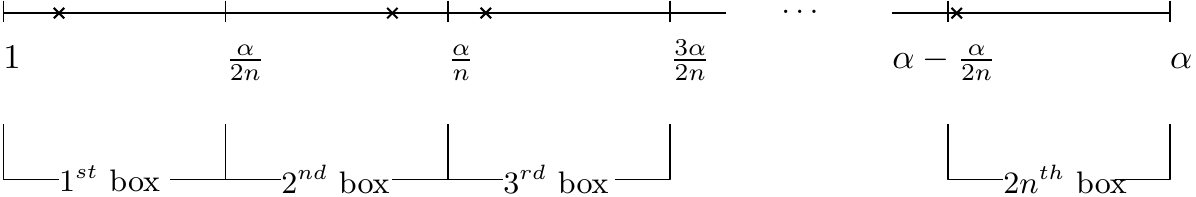}
		\caption{Splitting the time from $1$ to $\alpha$ into $2n$ boxes to show the existence of at least one label per box, for every edge, asymptotically almost surely.}
		\label{fig:boxes}
	\end{figure}
	
	We will show that for every edge of $G_s$, there will be asymptotically almost surely at least one of its labels that falls in the first box, one of its labels that falls in the second box, etc.
	But first, let us note the following:
	\begin{observation*}
		If for every edge $e\in E$ and for every box $B_i$ there is at least one label of $e$ that lies within $B_i$, then there exists an exploration of $G_s(\alpha,k)$.
	\end{observation*}
	\begin{proof}
		Assume that for every edge $e\in E$ and for every box $B_i$ there is at least one label of $e$ that lies within $B_i$. Fix an arbitrary order $e_1,e_2,\ldots,e_{n-1}$ of the edges of $G_s$. Explore $e_1$ using its label that lies within $B_1$ to enter and its label that lies within $B_2$ to exit, explore $e_2$ using its label that lies within $B_3$ to enter and its label that lies within $B_4$ to exit, and so on and so forth.
	\end{proof}
	Note now that for a particular edge $e \in E$ and a particular box $B_i$ of $e$, the probability that $B_i$ contains none of the labels of $e$ is:
	\[Pr[B_i \text{ is empty}] =\left( 1- \frac{\frac{\alpha}{2n}}{\alpha}\right)^k \leq \left( 1- \frac{\frac{\alpha}{2n}}{\alpha}\right)^{6n \ln{n}}  \leq \frac{1}{n^{3}}.\]
	So the probability that there is an empty box of $e$ is:
	\[Pr[\text{there is an empty }B_i \text{ of } e] \leq 2n\cdot\frac{1}{n^3} = \frac{2}{n^2}, \]
	and so the probability that there exists an edge with an empty box is:
	\[Pr[\text{there is an edge with an empty box}] \leq \# \text{edges} \cdot \frac{2}{n^2} \leq \frac{2}{n}.\]
	Finally, the probability that we can explore all edges of $G_s(\alpha,k)$ is:
	\[Pr[\text{exploration}] \geq 1-\frac{2}{n}  \rightarrow 1 \text{, as } n \rightarrow +\infty\]
	The latter completes the proof of the theorem.
\end{proof}

\begin{theorem}
	If $\alpha \geq 4$ and $k=2$, then the probability that we can explore all edges of $G_s(\alpha,k)$ tends to zero as $n$ tends to infinity.
\end{theorem}
\begin{proof}
	We introduce the following definition, needed for the proof.
	\begin{definition}
		Let $e_1,e_2$ be two edges of a uniform random temporal star $G_s(\alpha,2)$, $\alpha \geq 2$. Let the labels of $e_1$ be $a_1,a_2$, with $a_1\leq a_2$. Let the labels of $e_2$ be $b_1,b_2$, with $b_1\leq b_2$. We say that $e_1,e_2$ are a \emph{blocking pair} (with respect to exploration in $G_s(\alpha,2)$) if $a_1 \leq b_1 \leq a_2 \leq b_2$, or $a_1 \leq b_1 \leq b_2 \leq a_2$, or $b_1  \leq a_1 \leq b_2 \leq a_2$, or $b_1 \leq a_1 \leq a_2 \leq b_2$.
	\end{definition}
	
	Consider two particular edges $e_1,e_2$ of $G_s(\alpha,2)$, $\alpha\geq 4$. 
	Let $\mathcal{E}$ be the event that $e_1,e_2$ are a blocking pair and $\mathcal{E'}$ be the event that $e_1,e_2$ have $4$ distinct labels in total. Then, the probability that $e_1,e_2$ are a blocking pair is:
	\begin{equation}\label{eq:block1}
	Pr[\mathcal{E}] = Pr[\mathcal{E}~|~\mathcal{E'}] \cdot Pr[\mathcal{E'}] + Pr[\mathcal{E}~|~\bar{\mathcal{E'}}] \cdot Pr[\bar{\mathcal{E'}}] \geq Pr[\mathcal{E}~|~\mathcal{E'}] \cdot Pr[\mathcal{E'}],
	\end{equation}
	where $\bar{\mathcal{E'}}$ denotes the complement of $\mathcal{E'}$.
	Note that if all labels of $e_1,e_2$ are distinct, then the probability that $e_1,e_2$ are a blocking pair is exactly the ratio of the ``good'' arrangements of the $4$ distinct labels, i.e.~those where $a_1 < b_1 < a_2 < b_2$, or $a_1 < b_1 < b_2 < a_2$, or $b_1  < a_1 < b_2 < a_2$, or $b_1 < a_1 < a_2 < b_2$, over the total number of possible arrangements of the $4$ distinct labels. So, equation~\ref{eq:block1} becomes:
	\begin{equation}\label{eq:block2}
	Pr[\mathcal{E}] \geq \frac{4}{4!} \cdot Pr[\mathcal{E'}] = \frac{1}{6} \cdot Pr[\mathcal{E'}] .
	\end{equation}
	
	Now, the probability that all $4$ labels of $e_1,e_2$ are distinct is:
	
	\begin{equation}\label{eq:block3}
	Pr[\mathcal{E'}] = \left( 1- \frac{1}{\alpha} \right) \cdot \left( 1- \frac{2}{\alpha} \right) \cdot \left( 1- \frac{3}{\alpha} \right) \geq \frac{3}{4} \cdot \frac{2}{4} \cdot \frac{1}{4} = \frac{3}{32}.
	\end{equation}
	
	
	Therefore, by equation~\ref{eq:block3}, equation~\ref{eq:block2} becomes $Pr[\mathcal{E}] \geq \displaystyle \frac{1}{64}$.
	
	So, we have:
	
	\[Pr[e_1,e_2 \text{ are not a blocking pair}] \leq \frac{63}{64}.\]
	Let us now arbitrarily group all edges of $G_s(\alpha,2)$ into $\lfloor \frac{n-1}{2} \rfloor$ independent pairs (with the possibility of an edge remaining unpaired). If there is an exploration in $G_s(\alpha,2)$, then there are no blocking pairs of edges in any such pairing and, thus, in the particular pairing we have chosen. So, the probability that we can explore all edges is:
	\begin{eqnarray*}
		Pr[\text{exploration}] &\leq& Pr[\text{no blocking pair exists in the group}]\\
							   &\leq& \left( \frac{63}{64} \right)^{\lfloor \frac{n-1}{2} \rfloor} \rightarrow 0 \text{, as } n \rightarrow +\infty
	\end{eqnarray*} 	
\end{proof}


Figure~\ref{fig:random_param} shows the current state of what is known for the explorability of $G_s(\alpha, k)$ depending on the values of $\alpha$ and $k$.
\begin{figure}[h]
	\centering
	\includegraphics[width=0.5\textwidth]{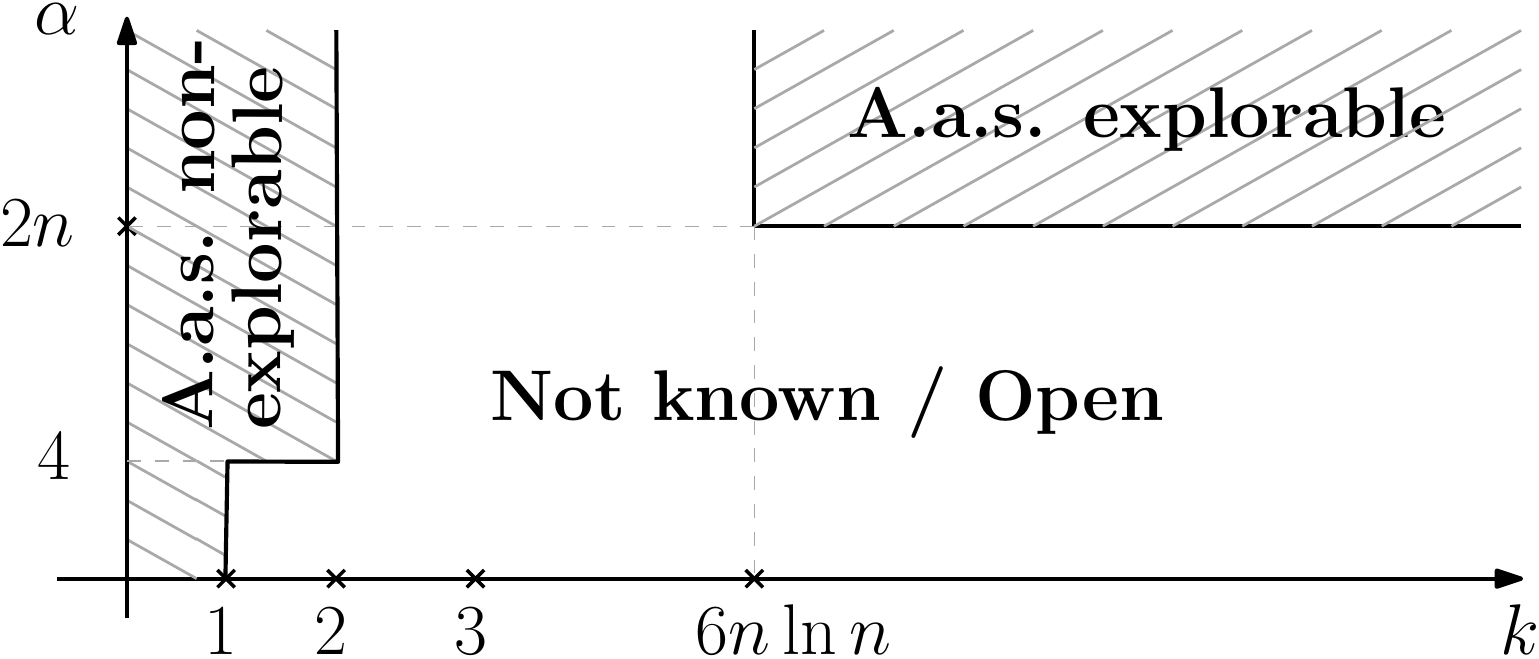}
	\caption{The shaded areas of the chart indicate the pairs $(\alpha,k)$ for which $G_s(\alpha,k)$ 
is asymptotically almost surely (a.a.s.) explorable and non-explorable, respectively.}
	\label{fig:random_param}
\end{figure}

\section{Conclusions and open problems}

In this paper, we have thoroughly investigated the computational complexity landscape of the temporal star exploration problems \SE\ and \MSE, depending on the maximum number $k$ of labels allowed per edge. 

We have shown that an optimal solution to the maximization problem~\textsc{MaxStarExp(2)}, on instances every edge of which has two labels per edge, can be efficiently found in $O(n\log n)$ time. This immediately implies that the decision version,~\textsc{StarExp(2)}, can be also solved in the same time.
We have proven that \textsc{StarExp(3)} can be solved in $O(n\log n)$ time, by carefully reducing it to instances of \textsc{2SAT} with number of clauses that is linear on the number of variables. This requires a more sophisticated analysis than what is needed to reduce \textsc{StarExp(3)} to an arbitrary \textsc{2SAT} instance; the latter would solve \textsc{StarExp(3)} in $O(n^2)$ time. 
For every $k\geq 6$, we show that \SE\ is NP-complete and \MSE\ is APX-complete. Indeed, we also give a greedy 2-approximation algorithm for \MSE.
Finally, we study the problem of exploring uniform random temporal stars whose edges have $k$ random labels (chosen uniformly at random within an interval $[1,\alpha]$, for some $\alpha\in\mathbb{N}$). We partially characterize the classes of uniform random temporal stars which, asymptotically almost surely, admit a complete (resp.~admit no complete) exploration. In particular, the ``blocking pairs'' technique used to show that there is asymptotically almost surely no complete exploration for $k=2$ and $\alpha \geq 4$ cannot be easily extended to large $k$. So, it remains open to determine the explorability of uniform random temporal stars for values of $k$ between $2$ and $6n\ln{n}$.

We pose here a question regarding the complexity of the maximization problem \textsc{MaxStarExp(3)}, which remains an open problem, as well as the complexity of \SE\ and \MSE, for $k\in \{4,5\}$. An interesting variation of \SE\ and \MSE\ is the case where the consecutive labels of every edge are $\lambda$ time steps apart, for some $\lambda \in \mathbb{N}$. What is the complexity and/or best approximation factor one may hope for in this case?

\bibliographystyle{plainurl}
\bibliography{new_bib}

\end{document}